\documentclass{article}
\usepackage[utf8]{inputenc}
\usepackage[margin=4cm]{geometry}

\title{Notes on the Sum-Rank Weight of a Matrix with Restricted Rank}
\author{Hugo Sauerbier Couv\'ee, Hedongliang Liu}
\date{Technical University of Munich\\[2ex]\today}

\usepackage{xcolor}

\usepackage{titlesec}

\titleformat{\chapter}[display]
  {\normalfont\bfseries}{}{0pt}{\Huge}

\usepackage{calrsfs}

\usepackage{todonotes}

\usepackage[ruled,linesnumbered]{algorithm2e}
\SetKwComment{Comment}{// }{}
\DontPrintSemicolon

\usepackage{amsmath, amsfonts, amssymb, amsthm, stmaryrd, extarrows, hyperref}
\newtheorem{theorem}{Theorem}[]

\newtheorem{corollary}{Corollary}
\newtheorem{lemma}{Lemma}

\theoremstyle{remark}

\newcommand{\F}{\mathbb{F}}
\newcommand{\N}{\mathbb{N}}

\newcommand{\Gr}{\operatorname{Gr}}

\usepackage{comment}

\usepackage{amsmath, amssymb, amsbsy}
\usepackage{nicematrix}
\usepackage{graphicx}
\usepackage[capitalize]{cleveref}
\definecolor{TUMBlue}{RGB}{0,101,189} % Pantone 300 (0, 0.396, 0.7412)
\definecolor{TUMBlueDark}{RGB}{0,82,147} % Pantone 301 (0, 0.3216,0.5765)
\definecolor{TUMBlueLight}{RGB}{152,198,234} % Pantone 283 (0.596,0.776,0.918)
\definecolor{TUMBlueMedium}{RGB}{100,160,200} % Pantone 542 (0.392, 0.627, 0.784)

\definecolor{TUMIvory}{RGB}{218,215,203} % Pantone 7527 -Elfenbein
\definecolor{TUMGreen}{RGB}{162,173,0} % Pantone 383 - Grün (0.6353,0.6784,0)
\definecolor{TUMGray}{gray}{0.6} % Grau 60%

\definecolor{TUMOrange}{RGB}{227,114,34} % Pantone 158 - Orange (0.8902, 0.4471, 0.133)
\definecolor{TUMGreenDark}{RGB}{0,124,48} % (0,0.4863,0.1882)
\definecolor{TUMRed}{RGB}{196,7,27} % (0.7686,0.02745,0.10588)

\definecolor{TUMPink}{RGB}{181,92,165}
\definecolor{TUMPinkDark}{RGB}{155,70,141}
\definecolor{TUMPink1}{RGB}{198, 128, 187}
\definecolor{TUMPink2}{RGB}{214, 164, 206}
\definecolor{TUMPink3}{RGB}{230, 199, 225}
\definecolor{TUMPink4}{RGB}{246, 234, 244}
\def\ve#1{{\mathchoice{\mbox{\boldmath$\displaystyle #1$}}%
              {\mbox{\boldmath$\textstyle #1$}}%
              {\mbox{\boldmath$\scriptstyle #1$}}%
              {\mbox{\boldmath$\scriptscriptstyle #1$}}}}

\newcommand{\Fq}{\ensuremath{\mathbb{F}_q}}

\newcommand{\bbN}{\ensuremath{\mathbb{N}}}

\newcommand{\bn}{\ve{n}}

\newcommand{\bA}{\ve{A}}

\newcommand{\bM}{\ve{M}}

\newcommand{\sfc}{\mathsf{c}}

\newcommand{\myspan}[1]{\ensuremath{\left\langle #1\right\rangle}}

\DeclareMathOperator{\rank}{rank}
\newcommand{\wtSR}[1]{\ensuremath{\mathrm{wt}_{\mathsf{SR},#1}}}

\rmfamily\mdseries\scshape

\begin{document}

\maketitle
%\tableofcontents

%\newpage
\begin{abstract}
These notes cover a few calculations regarding the sum-rank weight of a matrix in relation to its rank. In particular, a formula and lower bounds are given on the probability that a matrix of rank $t$ consisting of $\ell$ blocks has sum-rank weight $\ell t$.
\end{abstract}
% \newpage
\section{Notations and Preliminaries}

For an $a\in\bbN$, denote $[a]:=\{1,2,\dots, a\}$.
Throughout this note we consider matrices over the finite field $\F_q$, where $q$ is a prime power.\\
\\
For $t \in \N$, we define
$Q_t(x) := (x-1)(x-q)\cdots(x-q^{t-1}) = \prod_{i=0}^{t-1}(x-q^i)$. Numbers of the form $Q_t(q^r)$ show up often when calculating the number of subspaces and matrices with certain properties, see also \cite{matrices_finite},\cite{gadouleau},\cite{counting_matrices},\cite{Landsberg1893},\cite{migler2004weight},\cite{WANG2010297}.
  % \end{align*}
% \end{definition}
Denote by $\Gr_{t}(\Fq^m)$ the \emph{Grassmanian} (a set of all $t$-dim subspaces of $\Fq^m$). The cardinality of the Grassmanian $|\Gr_t(\Fq^m)|$ is denoted by $\begin{bmatrix}m\\t\end{bmatrix}_q$, which is often called the \emph{Gaussian binomial} or \emph{$q$-binomial}.
% \section{Applications of $Q_t(q^m)$ in counting matrices}
\begin{lemma}
  \label{lem:num-r-dim-subspaces}  Let $t,m\in\bbN$ and $m \geq t$. Then
\[
\begin{bmatrix}m\\t\end{bmatrix}_q = \frac{Q_t(q^m)}{Q_t(q^t)}.
\]
\end{lemma}
\begin{proof}
Here, $Q_t(q^m) = (q^m-1)(q^m-q)\cdots(q^m-q^{t-1})$ gives the total number of ordered bases consisting of $t$ linearly independent vectors of $\F_q^m$. This overcounts the number of $t$-dim subspaces, so we have to divide it by $Q_t(q^t)$, the number of different ordered bases that span the same fixed $t$-dim subspace, which is also the number of invertible $t \times t$ matrices.
\end{proof}
\newpage
\begin{lemma}
  \label{lem:num-matrices-fixed-column-space}
  Fix an $t$-dim subspace $V \subseteq \F_q^m$. Then the number of $m\times n$ matrices with column space equal to $V$ is given by
\begin{align*}
\#\{\text{\emph{$m \times n$ matrices with colsp. $V$}}\} &= \#\{\text{\emph{lin. maps $\F_q^n \to \F_q^m$ with image $V$}}\}\\
&= \#\{\text{\emph{surjective lin. maps $\F_q^n \to V \cong \F_q^t$}}\}\\
&= \#\{\text{\emph{$t \times n$ matrices with rank $t$}}\}\\
&= \#\text{\emph{ways to choose $t$ lin. indep. vectors from $\F_q^n$ (ordered)}}\\
&= Q_t(q^n).
\end{align*}
\end{lemma}

\begin{lemma}\label{lem:number-matrices-rank-r}
Let $t,m,n \in\bbN$ and $m\geq t, n\geq t$. Then
\begin{align*}
\#\{\text{\emph{$m \times n$ matrices of rank $t$ over $\Fq$}}\} & = \#\{\text{\emph{$t$-dim subspaces of $\F_q^m$}}\} \\
&\text{}  \quad \cdot\;\#\{\text{\emph{$m \times n$ matrices whose colsp. is a fixed $t$-dim subspace}}\} \\
&= \begin{bmatrix}m\\t\end{bmatrix}_q {Q_t(q^n)}\\
&= \frac{Q_t(q^m)Q_t(q^n)}{Q_t(q^t)}.
\end{align*}
\end{lemma}

\section{Sum-rank weight of a matrix with restricted rank}
In this section, we discuss about the \emph{sum-rank weight} of a matrix with restricted rank.

Let $\ell,n \in\bbN$ and $\bn_\ell= (n_1, \dots,n_\ell)\in\bbN^\ell$ be an ordered partition of $n=\sum_{i=1}^\ell n_i$.
% \hsc{Define ordered partition}
For a matrix %with $n$ columns and it is ordered partitioned into $\ell$ blocks column-wisely, e.g.,
\begin{align*}
  \bA =
  \begin{pmatrix}
    \bA_1 & \bA_2 & \dots & \bA_\ell
  \end{pmatrix}
                        \in\Fq^{m\times n}
  %                       \quad\quad\text{ or }\quad\quad
  %                       \bB =
  %                       \begin{pmatrix}
  %                         \bB_1\\
  %                         \bB_2\\
  %                         \vdots\\
  %                         \bB_\ell
  %                       \end{pmatrix}
  % \in\Fq^{n\times m}
\end{align*}
% \aw{Before that, matrices were bold -- unify}
where $\bA_i\in\Fq^{m\times n_i}$.
% and each $\bB_i\in\Fq^{n_i\times m}$, we say $\bA$ has a \emph{column-wise} ordered partition with respect to $\bn_{\ell}$ and $\bB$ has a \emph{row-wise} ordered partition w.r.t.~$\bn_{\ell}$.
The sum-rank weight w.r.t.~$\bn_\ell$ is
\begin{align*}
  \wtSR{\bn_\ell}(\bA)= \sum_{i=1}^{\ell} \rank(\bA_i)
  % \quad\quad\text{ or }\quad\quad
  % \wtSR{\bn_\ell}(\bB)= \sum_{\indBlock=1}^{\ell} \rank(\bB_\indBlock)
\end{align*}
We can find the following relation between the sum-rank weight and the rank of a matrix.
\begin{lemma}
  \label{lem:rank-leq-sumrank-mat}
  For a matrix $\bA\in\Fq^{m\times n}$ and an ordered partition $\bn_\ell= (n_1, \dots,n_\ell)$ of $n$, $\rank(\bA)\leq\wtSR{\bn_{\ell}}(\bA)\leq \ell\cdot \rank(\bA)$.
  % Similarly, for a matrix $\bB\in\Fq^{n\times m}$, $\rank(\bB)\leq \wtSR{\bn_{\ell}}(\bB)\leq \ell\cdot \rank(\bB)$.
\end{lemma}
\begin{proof}
  Denote by $\myspan{\bA}_{\sfc}$ the column space of a matrix $\bA$. For the first inequality,
  \begin{align*}
    \rank(\bA) = \dim(\myspan{\bA}_{\sfc}) &= \dim(\myspan{\bA_1}_{\sfc}+\cdots + \myspan{\bA_\ell}_{\sfc})\\
                                       &\leq \dim(\myspan{\bA_1}_{\sfc}) + \cdots \dim(\myspan{\bA_\ell}_{\sfc})\\
    &=\rank(\bA_1) +\cdots +\rank(\bA_\ell) = \wtSR{\bn_\ell}(\bA)\ .
  \end{align*}
  For the second inequality,
  \begin{align*}
    \wtSR{\bn_\ell}(\bA)= \sum_{i=1}^{\ell}\rank(\bA_i) % +\cdots +\rank(\bA_\ell)
    \leq \sum_{i=1}^{\ell}\rank(\bA) =\ell\cdot \rank(\bA).
  \end{align*}
  % For the matrix $\bB\in\Fq^{n\times m}$ with row-wise ordered partition, the proof is similar with considering the row space of $\bB$ and $\bB_i$'s.
\end{proof}
\subsection{Probability of a rank-restricted matrix attaining the maximal sum-rank weight}
In the following, we consider a matrix $\bA\in\Fq^{m\times n}$ with restricted rank, $\rank(\bA)=t$.
The goal is to derive an expression on the probability of $\wtSR{\bn_\ell}(\bA)=\ell\cdot \rank(\bA)$ w.r.t.~an ordered partition $\bn_\ell= (n_1, \dots,n_\ell)$ of $n$ where $m\geq t$ and $n_i\geq t,\forall i\in[\ell]$.

% Recall that $Q_t(q^n)\coloneq \prod_{i=0}^{t-1}(q^n-q^i)$.
\begin{lemma}
  \label{lem:minimal-ordered partition}
  Let $m,n,t,\ell\in\bbN$ and $q$ be a prime power. For any ordered partition $\bn_\ell= (n_1, \dots,n_\ell)$ of $n$ where $m\geq t$ and $n_i\geq t,\forall i\in[\ell]$ we have
  \begin{align*}
    Q_t(q^{n_1})Q_t(q^{n_2})\cdots Q_t(q^{n_\ell})\geq Q_t(q^t)^{\ell-1}Q_t(q^{n-(\ell-1)t})\ .
  \end{align*}
  In other words, $(\underbrace{t, t,\dots,t}_{\ell-1}, n-(\ell-1)t)$ (or any permutation) is an ordered partition of $n$ that minimizes $\prod_{i=1}^{\ell}Q_t(q^{n_i})$.
\end{lemma}
\begin{proof}
  % Note that $\prod_{i=1}^{\ell}Q_t(q^{n_i})$ is invariant under permutation of a fixed $\bn_{\ell}$. W.l.o.g., assume $n_1\leq n_2\leq \cdots\leq n_{\ell}$.
  For $\ell=1$ the statement holds trivially.
  % For $\ell\geq 2$, we prove the statement by induction.

  For $\ell=2$, for any $a,b\in\bbN, a+b=n, a\geq b\geq t+1$, and $i=1,\dots, t-1$, we have
  \begin{align*}
    &(q^a-q^i)(q^b-q^i) - (q^{a+1}-q^i)(q^{b-1}-q^i) \\
    =& (q^{a+b}-q^{i+a}-q^{i+b}+q^{2i}) - (q^{a+b}-q^{i+a+1}-q^{i+b-1}+q^{2i})\\
    =& q^{i+a+1} -q^{i+a} + q^{i+b-1} - q^{i+b}\\
    =& q^{i+b-1}((q^{a-b+2}-q^{a-b+1}) - (q-1))\\
    =& q^{i+b-1} (q-1) (q^{a-b+1} - 1) > 0\ ,
  \end{align*}
  The last inequality holds because $q>1$ and $a \geq b$. This shows that 
  \[
    Q_t(q^a)Q_t(q^b) > Q_t(q^{a+1})Q_t(q^{b-1}).
  \]
Then for any $a,b\in\bbN, a+b=n, a\geq t, b\geq t$, we have
  % and thus by repeating this inequality $t$\lia{Should this be $b-t$?} times we obtain
      \begin{align}
        \label{eq:inequality-ell=two}
        Q_t(q^a)Q_t(q^b) \geq Q_t(q^{n-t})Q_t(q^{t}),
        \end{align}
  where the equality holds if and only if $a=t$ or $b=t$.

  Now assume that the statement holds for all $\ell \in \{2,\ldots, \ell'\}$ for some integer $\ell' \geq 2$. We will show that the statement also holds for $\ell=\ell'+1$. 
  Note that $\sum_{i=1}^{\ell'} n_i=n-n_{\ell}$. By the assumption, we have
  \begin{align}
    Q_t(q^{n_1})Q_t(q^{n_2})\cdots Q_t(q^{n_{\ell'}})Q_t(q^{n_{\ell}}) & \geq Q_t(q^t)^{\ell'-1}Q_t(q^{n-n_\ell-(\ell'-1)t})Q_t(q^{n_{\ell}}) \nonumber \\
    & \geq Q_t(q^t)^{\ell'-1} Q_t(q^{n-n_\ell+(n_{\ell}-t)-(\ell'-1)t}) Q_t(q^{t}) \label{eq:apply-ell=two}\\
    & =  Q_t(q^t)^{\ell'} Q_t(q^{n-\ell't})\nonumber\\
    & =  Q_t(q^t)^{\ell-1} Q_t(q^{n-(\ell-1)t})\nonumber \ ,
  \end{align}
  where \eqref{eq:apply-ell=two} follows from \eqref{eq:inequality-ell=two} and by induction the statement holds for all $\ell\in\bbN$.
  % Now suppose the lemma holds for all $\ell \leq \ell'$, i.e. suppose that
%   \[
% Q_t(q^{n_1})Q_t(q^{n_2})\cdots Q_t(q^{n_{\ell'}})\geq Q_t(q^t)^{\ell'-1}Q_t(q^{(n_1+n_2+\ldots+n_{\ell'})-(\ell'-1)t})\ .
%   \]
%   Then
%   \begin{align*}
% Q_t(q^{n_1}) Q_t(q^{n_2})\cdots Q_t(q^{n_{\ell'}}) Q_t(q^{n_{\ell'+1}}) &\geq 
% Q_t(q^t)^{\ell'-1} 
% Q_t(q^{n_{\ell'+1}})
% Q_t(q^{(n_1+n_2+\ldots+n_{\ell'})-(\ell'-1)t}) \\
% &\geq Q_t(q^t)^{\ell'-1} 
% Q_t(q^{t})
% Q_t(q^{(n_1+n_2+\ldots+n_{\ell'})-(\ell'-1)t + n_{\ell'+1} -t})\\
% &= Q_t(q^t)^{\ell'} 
% Q_t(q^{(n_1+n_2+\ldots+n_{\ell'+1})-(\ell')t})\\
%   \end{align*}
%   so the lemma also holds for $\ell = \ell' + 1$. By strong induction the lemma now follows for all $\ell \in \N$. 
\end{proof}

\begin{theorem}
\label{thm:prob-full-sum-rank}
Let $m,n,t,\ell\in\bbN$ and $q$ be a prime power. Let $\bn_\ell= (n_1, \dots,n_\ell)$ be an ordered partition of $n$ where $m\geq t$ and $n_i\geq t,\forall i\in[\ell]$. 
If $\bA$ is randomly drawn from the set $\{\bM \in \Fq^{m\times n}\ |\ \rank(\bM)=t\}$ (with uniform probability distribution), then 
\begin{align*}
    \Pr[\wtSR{\bn_{\ell}}(\bA)=\ell t\ |\ \rank(\bA)=t]     
    % &= \frac{Q_t(q^{n_1})Q_t(q^{n_2})\cdots Q_t(q^{n_\ell})}{Q_t(q^n)}\\
     &\geq  %\frac{Q_t(q^t)^{\ell-1}Q_t(q^{n-(\ell-1)t})}{Q_t(q^n)} =  
     \prod_{i=1}^{t-1}\frac{(1-q^{i-t})^{\ell - 1} (1-q^{i-n+(\ell-1)t})}{1-q^{i-n}}\ .
     % &> \prod_{i=1}^{t-1}(1-q^{i-t})^{\ell - 1} (1-q^{i-n+(\ell-1)t})\\
     % &> \prod_{j=1}^{\infty}(1-q^{-j})^{\ell} = (\frac{1}{q};\frac{1}{q})_{\infty}^{\ell},
  \end{align*}
  The equality holds if and only if $\bn_{\ell}=(t, t,\dots,t, n-(\ell-1)t)$ (or any permutation).
\end{theorem}
\begin{proof}
    Given the restriction $\rank(\bA) = t$, $\wtSR{\bn_{\ell}}(\bA)=\ell t$ holds if and only if the column space of $\bA_i$ coincides with the column space of $\bA$, for all $i\in[\ell]$.
  Hence, the number of matrices $\bA$ such that $\rank(\bA_i)=t, \forall i\in[\ell]$ is
  \begin{align*}
    \frac{Q_t(q^m)}{Q_{t}(q^t)}\cdot Q_t(q^{n_1})Q_t(q^{n_2})\cdots Q_t(q^{n_\ell}),
  \end{align*}
  where $\frac{Q_t(q^m)}{Q_{t}(q^t)}$ is the number of $t$-dim subspaces of $\Fq^m$ (\cref{lem:num-r-dim-subspaces}) and $Q_t(q^{n_i})$ is the number of $m\times n_i$ matrices with column space equal to a fixed $t$-dim subspace $V\subseteq\Fq^m$ (\cref{lem:num-matrices-fixed-column-space}). 

  Recall from \cref{lem:number-matrices-rank-r} that the number of $m\times n$ matrices of rank $t$ is
  \begin{align*}
    \frac{Q_t(q^m)Q_t(q^n)}{Q_t(q^t)}\ .
  \end{align*}
  Then,
  \begin{align*}
    \Pr[\wtSR{\bn_{\ell}}(\bA)=\ell t\ |\ \rank(\bA)=t] &= \frac{\frac{Q_t(q^m)}{Q_{t}(q^t)}\cdot Q_t(q^{n_1})Q_t(q^{n_2})\cdots Q_t(q^{n_\ell})}{\frac{Q_t(q^m)Q_t(q^n)}{Q_t(q^t)}}\\
    &= \frac{Q_t(q^{n_1})Q_t(q^{n_2})\cdots Q_t(q^{n_\ell})}{Q_t(q^n)}\\
   %  =& \prod_{j=1}^{t-1}\frac{\prod_{i=1}^{\ell} (q^{n_i}-q^j)}{q^{n}-q^j}
     &\geq  \frac{Q_t(q^t)^{\ell-1}Q_t(q^{n-(\ell-1)t})}{Q_t(q^n)}\\
     &=  \prod_{i=1}^{t-1}\frac{(q^{t}-q^i)^{\ell - 1} (q^{n-(\ell-1)t}-q^i)}{q^{n}-q^i}\\
     &=  \prod_{i=1}^{t-1}\frac{(1-q^{i-t})^{\ell - 1} (1-q^{i-n+(\ell-1)t})}{1-q^{i-n}}\ ,
  \end{align*}
  where the inequality follows from \cref{lem:minimal-ordered partition}.
\end{proof}
Denote by $(a;x)_{\infty}$ the $q$-Pochhammer symbol $\prod_{j=0}^{\infty}(1-ax^j)$. It can be shown that $(\frac{1}{q};\frac{1}{q})_{\infty}$ converges quickly to 1 as $q \to \infty$, e.g.,
\[
 (\tfrac{1}{2};\tfrac{1}{2})_{\infty} \approx 0.289,\quad  (\tfrac{1}{3};\tfrac{1}{3})_{\infty} \approx 0.560,\quad (\tfrac{1}{7};\tfrac{1}{7})_{\infty} \approx 0.837,\quad (\tfrac{1}{13};\tfrac{1}{13})_{\infty} \approx 0.917,\quad \text{etc.}
\]
Furthermore, by Euler's pentagonal theorem, we have 
\begin{align*}
(\tfrac{1}{q};\tfrac{1}{q})_{\infty} = \prod_{i=1}^{\infty} \left(1-\left(\frac{1}{q}\right)^i\right) = 1-\frac{1}{q}-\frac{1}{q^2}+\frac{1}{q^5}+\frac{1}{q^7}-\frac{1}{q^{12}}-\frac{1}{q^{15}}+\dots.
\end{align*}
Then, for $q \geq 2$, we have
% \begin{align*}
$(\frac{1}{q};\frac{1}{q})_{\infty} > 1-\frac{1}{q}-\frac{1}{q^2} $ . %$> 1 - \frac{1}{q-1}$.
% \end{align*}
% \todo{Is it relevant to prove this?} \lia{Is there some reference already proving this?}
\begin{corollary}
    Following the same notation as \cref{thm:prob-full-sum-rank}, we have
    % Let $m,n,t, \ell\in\bbN$ and $q$ be a prime power. Let $\bn_\ell= (n_1, \dots,n_\ell)$ be an ordered partition of $n$ where $m\geq t$ and $n_i\geq t,\forall i\in[\ell]$. 
% If $\bA$ is randomly drawn from the set $\{\bM \in \Fq^{m\times n}\ |\ \rank(\bM)=t\}$ (with uniform probability distribution), then 
\begin{align*}
    \Pr[\wtSR{\bn_{\ell}}(\bA)=\ell t\ |\ \rank(\bA)=t]     
     % =  \prod_{i=1}^{t-1}\frac{(1-q^{i-t})^{\ell - 1} (1-q^{i-n+(\ell-1)t})}{1-q^{i-n}}\\
     % &> \prod_{i=1}^{t-1}(1-q^{i-t})^{\ell - 1} (1-q^{i-n+(\ell-1)t})\\
     % &> \prod_{j=1}^{\infty}(1-q^{-j})^{\ell} = (\frac{1}{q};\frac{1}{q})_{\infty}^{\ell},
     &> \left(1-\frac{1}{q}-\frac{1}{q^2}\right)^{\ell}\ .
  \end{align*}
  In particular, if $q>\ell$, i.e., $\ell\leq q-1$, we have
  \begin{align*}
      \Pr[\wtSR{\bn_{\ell}}(\bA)=\ell t\ |\ \rank(\bA)=t] >\frac{1}{4}\ .
  \end{align*}
\end{corollary}
\begin{proof}
    It follows from \cref{thm:prob-full-sum-rank} that
    \begin{align*}
    \Pr[\wtSR{\bn_{\ell}}(\bA)=\ell t\ |\ \rank(\bA)=t]     
     &\geq \prod_{i=1}^{t-1}\frac{(1-q^{i-t})^{\ell - 1} (1-q^{i-n+(\ell-1)t})}{1-q^{i-n}}\\
     &> \prod_{i=1}^{t-1}(1-q^{i-t})^{\ell - 1} (1-q^{i-n+(\ell-1)t})\\
     &> \prod_{j=1}^{\infty}(1-q^{-j})^{\ell} = (\tfrac{1}{q};\tfrac{1}{q})_{\infty}^{\ell} > \left(1-\frac{1}{q}-\frac{1}{q^2}\right)^{\ell}.
    \end{align*}
    If $\ell\leq q-1$, we have $\left(1-\frac{1}{q}-\frac{1}{q^2}\right)^{\ell} \geq \left(1-\frac{1}{q}-\frac{1}{q^2}\right)^{q-1}$, which increases as $q$ increases. Therefore,
    \begin{align*}
        \Pr[\wtSR{\bn_{\ell}}(\bA)=\ell t\ |\ \rank(\bA)=t] > \left(1-\frac{1}{q}-\frac{1}{q^2}\right)^{q-1} \Big|_{q=2} = \frac{1}{4}\ .
    \end{align*}
\end{proof}

%\nocite{*}

\bibliographystyle{plain}
\bibliography{biblio}

\end{document}